\documentclass[11pt]{article}

\usepackage[margin=0.78in]{geometry}
\usepackage[T1]{fontenc}
\usepackage{lmodern}
\usepackage{microtype}
\usepackage{amsmath,amssymb,amsthm,mathtools}
\usepackage{booktabs}
\usepackage[hidelinks]{hyperref}
\usepackage{setspace}
\usepackage{natbib}
\usepackage{graphicx}

\newtheorem{theorem}{Theorem} 
\newtheorem{proposition}{Proposition}
\newtheorem{corollary}{Corollary}

\newcommand{\E}{\mathbb E}

\newcommand{\1}{\mathbf 1}

\newcommand{\cK}{\mathcal K}

\newgeometry{margin=1.25in}

\title{Optimal Experimental Design and Estimation \\ when Potential Outcomes are Bounded}
\author{Peter Hull\thanks{Brown University Department of Economics; peter\_hull@brown.edu.\\ I thank P Aronow, Kirill Borusyak, Jimbo Brand, Jiafeng Chen, Paul Goldsmith-Pinkham, Evan Munro, and Jonathan Roth for several comments and conversations that helped shape this paper. OpenAI's ChatGPT also contributed valuable insights. Refine.ink was used to check for consistency and
clarity. }}
\date{August 2026}

\begin{document}
\maketitle

\begin{abstract}
\noindent I study the optimal design and analysis of randomized experiments for estimating finite-population average treatment effects when potential outcomes are known to be bounded, as with binary outcomes. Among all assignment mechanisms and a broad class of affine estimators, worst-case mean-squared error (MSE) is minimized by independent random assignment and an unconventional regression of the support-midpoint-centered outcome on the recentered treatment, with no intercept. This contrasts with the usual prescription of balanced complete randomization and difference-in-means estimation: when outcomes are bounded, randomness in the realized treatment share is informative. The worst-case gain over full-sample complete randomization is asymptotically small, but gains can be first-order relative to other designs: complete within-pair randomization and pair-fixed-effect regression have twice the worst-case MSE. I then consider the general problem, allowing for any measurable estimator. Independent random assignment remains optimal, and the generally-nonlinear optimal estimator can meaningfully reduce worst-case MSE.
\end{abstract}

\newpage
\onehalfspacing

\section{Introduction}

Complete randomization in experiments---where a fixed number of units are treated---is often viewed as an unambiguous efficiency improvement over independent random assignment. By fixing the treated share, it eliminates chance imbalances while preserving uniform randomness over unit labels. Recent minimax results formalize this intuition, selecting complete randomization and difference-in-means estimation of average treatment effects (ATEs) with unrestricted potential outcomes \citep{bai2023why,kallus2021optimality}.

This paper shows that the formal optimality of complete randomization depends on the researcher not having an \emph{a priori} bound on potential outcomes---as one has with binary or otherwise limited-support outcomes. I show that, given outcome bounds, complete randomization may no longer be minimax. Instead, worst-case mean-squared error (MSE) is minimized by independent random assignment---with a random treated share. Intuitively, when outcomes are unbounded complete randomization is preferred because variation in the treated share can be exploited by an adversary: making the level of potential outcomes arbitrarily large or small can blow up the estimator with even a small imbalance. By anchoring the level of potential outcomes, ex ante bounds prevent this and make a random treated share useful for ATE estimation. In particular, independent randomization avoids any pairwise correlation between treatment assignments which an adversary can exploit through the heterogeneity (rather than the level) of potential outcomes. 

I first show this reversal over a large class of affine estimators satisfying a natural condition, which ensures equivariance to midpoint-preserving shifts in treatment effects.\footnote{More precisely, the equivariance restriction says that increasing every treatment effect by some $t$ while holding each potential-outcome midpoint fixed must increase the estimate by exactly $t$.} This class includes simple or weighted difference-in-means, ordinary or weighted least squares with predetermined controls or fixed effects, and many other estimators.  I show independent randomization is minimax-optimal with an unconventional estimator: a regression of the support-midpoint-centered outcome on the recentered treatment with no intercept. Treatment recentering ensures design-based identification \citep{borusyakHull2023nonrandom}, while midpoint-centering the outcome makes the random treated share variation informative. I further show complete randomization is strictly suboptimal in this setting.

The worst-case MSE reduction of the optimal design and affine estimator is asymptotically small relative to full-sample complete randomization and the usual difference-in-means estimator, but it can be large relative to other popular strategies. Indeed, I show the procedure of paired randomization and pair-fixed-effect regression has twice the worst-case MSE by eliminating variation in treated shares within each pair. While the minimax analysis may be overly pessimistic for this case, as it doesn't incorporate  \emph{a priori} information on the possible similarity of potential outcomes within strata, it formalizes the sense in which repeated balance restrictions can be costly when the strata are not prognostically informative.

Finally, I consider  the general minimax problem which allows for any design and measurable estimator. I show that independent randomization remains minimax-optimal, while the optimal estimator is generally nonlinear. This estimator is characterized by a finite-dimensional convex program with quadratic constraints, and has a least-favorable-prior representation analogous to the  \cite{hodges1982minimax} estimator. It generally improves on the unbiased affine estimator by adding bias and reducing variance. Numerically, I find that using this optimal estimator decreases worst-case MSE by 15-30\% for moderate sample sizes, with the benefit declining asymptotically. I show how the expected MSE of different estimators, including optimal ones, can be computed in experiments with binary outcomes. 

This paper builds on several interconnected literatures. From statistics, a closely related literature studies optimal finite-population sampling and estimation of means. Early results study different restrictions on the population or the estimator class. \cite{godambe1955unified} shows that, for broad sampling designs and unrestricted population values, there is generally no unbiased linear estimator that uniformly minimizes variance.  \cite{godambeJoshi1965admissibility} establish the admissibility of the Horvitz–Thompson estimator within the class of design-unbiased estimators. \cite{bickelLehmann1981minimax} impose a bound on finite-population dispersion and show that the sample mean is minimax under simple random sampling. Most directly related is \cite{hodges1982minimax}, who show that if all finite-population values lie in a known interval then the minimax estimator under simple random sampling shrinks the sample mean toward the interval midpoint; moreover, simple random sampling paired with this estimator is minimax among all sampling and estimation strategies.\footnote{Other work develops minimax sampling theory for alternative parameter spaces and estimator classes, including \cite{joshi1979best}, \cite{chengLi1983minimax}, \cite{gabler1988conditional,gabler1990minimax}, and \cite{stenger1989asymptotic}.}  This paper adopts a similar logic for the estimation of causal effects.\footnote{\citet[][Lemma 4.1 and Proposition 4.2]{harshawEtAl2024balancing} provide a close antecedent. Fixing the Horvitz–Thompson estimator and restricting to designs with marginal treatment probability one-half, they show that worst-case MSE over an $\ell_2$ bound is minimized by independent randomization. This paper's Proposition \ref{prop1} strengthens this result by allowing arbitrary assignment mechanisms and optimizing jointly over a large set of affine estimators. Theorem \ref{theorem1} further generalizes by optimizing over all measurable estimators. }

The closest contemporary paper in this literature is \cite{aronowLopatto2026minimax}, who consider finite-population totals when each unit’s outcome has a known and potentially unit-specific support interval. Holding marginal inclusion probabilities fixed, they derive a sharp lower bound on the maximum MSE of any design-unbiased estimator and show it is attained precisely when the inclusion indicators are pairwise independent, in which case the optimal estimator is a midpoint-differenced Horvitz–Thompson estimator. They also jointly optimize the sampling design and design-unbiased estimator. My problem has the same bounded-support and midpoint-adjustment logic, but differs in two key ways. First, I study a causal setting where treatment assignment yields a nonstandard “one-from-each-pair” sampling from treated and untreated potential outcomes. Second, and more importantly, I impose no design-unbiasedness restriction and ultimately optimize over all estimators.\footnote{A related robust-design literature asks how randomization protects against misspecification or unknown dependence rather than bounded support. \cite{wu1981robustness} connects randomization to worst-case mean squared error protection in comparative experiments, while \cite{bickelHerzberg1979robustness} study designs robust to serially correlated errors. These papers provide an additional precedent for interpreting independent randomization as protection against an adversarial outcome configuration, although the settings are quite different.} 

More directly, this paper adds to the literature on optimal assignments and estimators for average treatment effects. Classic Neyman allocation fixes assigned shares to minimize the variance of differences in means given arm-specific outcome variances. Modern minimax results yield decision-theoretic foundations for related prescriptions. \cite{kallus2021optimality} shows complete randomization is minimax when the relevant conditional-mean class is permutation symmetric. \cite{bai2023why} studies joint optimization of assignment and linear estimation, obtaining difference-in-means at the Neyman allocation under his parameter class. My results are best viewed as complementing these analyses which leave outcome location unrestricted.\footnote{\cite{yamin2026when} also studies minimax experimental design with bounded outcomes, in a setting with pilot data and \emph{iid} sampling of potential outcomes. With difference-in-means estimation and fixed treatment counts, he shows balanced complete randomization is minimax absent informative pilot evidence; my analysis considers all estimators and designs while allowing for arbitrary dependence in fixed potential outcomes. My minimax analysis under bounds is also similar in spirit to \cite{deChaisemartin2024trading}, who assumes bounded strata-specific ATEs and derives the minimax linear combination of stratum-specific unbiased estimators. }  I show that a known outcome support can be exploited with a random treated share; conditional on the share, however, my optimal design is exactly the complete-randomization designs considered in earlier work.\footnote{Other work uses baseline information to improve on complete randomization: \cite{bai2022optimality} proves optimality of certain matched pairs among stratified designs that treat every unit with probability one half; \cite{hahnHiranoKarlan2011adaptive} and \cite{tabordMeehan2023stratification} use earlier experimental waves to choose treatment propensities and strata, and \cite{harshawEtAl2024balancing} construct a design that explicitly trades off balance against worst-case robustness.} 

A very close contemporary paper from this literature is \cite{sudijono2026sharp}, who independently study the unrestricted finite-population minimax problem and arrive at an equivalent characterization of the optimum under a change of parameterization. Their paper develops complementary aspects of the problem, deriving a sharp second-order expansion, an Airy-function characterization of the asymptotic minimax rule and least-favorable prior, and admissibility and dominance results for standard procedures. This paper instead emphasizes exact finite-sample results for natural affine estimator classes---including the result that complete randomization is minimax-suboptimal in this class, the factor-of-two worst-case cost of paired randomization, and the interpretation of the optimal procedures through midpoint centering and causal Hodges–Lehmann estimation.

Lastly, this paper builds on a recent literature on estimation with recentered estimators. \cite{borusyakHull2023nonrandom}  establish recentering for design-based identification of formula treatments, combining as-if-random assignments with other predetermined variables, while \cite{borusyakHull2026optimal} characterize efficient formula instruments for a given design. \cite{borusyakHullMunro2026robust} jointly optimize the design and recentered instrument under a minimax approximate-variance criterion. A motivating example in \cite{borusyakHullMunro2026robust}  shows that independent---rather than complete---randomization is minimax-optimal when the treatment formula is the assignment itself, using an asymptotic variance approximation and a particular class of well-behaved recentered instrumental variable estimators (see Section 2.3; see also Lemma 2 in \cite{borusyakHull2026optimal}). This paper derives a similar result for finite-sample MSE over a fully unrestricted class of estimators. 

The rest of this paper is organized as follows. The next section shows the initial result for the restricted affine estimator class. Section \ref{sec:general} then considers the most general problem. Section \ref{sec:conclusion} concludes. All proofs are given in the appendix. 

\section{Restricted Affine Estimators}\label{sec:affine}
Consider a fixed population of $N$ units. For each unit $i$, let $Y_i(0)$ and $Y_i(1)$ be fixed potential outcomes and let $D_i\in\{0,1\}$ denote treatment assignment. A researcher knows, prior to assignment, that potential outcomes are bounded: for known $L<U$,
\begin{align*}
Y_i(0),Y_i(1) \in [L,U],\hspace{0.3cm} i=1,\dots,N.
\end{align*}
For example, when studying a binary outcome the researcher knows $L=0$ and $U=1$. The parameter of interest is the finite-population ATE: 
\begin{align*}
\beta = \frac{1}{N}\sum_{i=1}^N (Y_i(1)-Y_i(0)).
\end{align*}

To estimate $\beta$, the researcher first chooses a design $\delta\in \Delta (\{0,1\}^N)$, where $\Delta(\cdot)$ denotes the simplex, along with an estimator. Treatment assignments are drawn from $\delta$ and determine outcomes $Y_i=Y_i(0)(1-D_i)+Y_i(1)D_i$. The researcher then applies the estimator to $(D_i,Y_i)_{i=1}^N$, yielding an estimate $\hat\beta$. Here we restrict to affine estimators, of the form:
\begin{align*}
\hat\beta = a(D)+b(D)^\prime Y,
\end{align*}
where $D$ and $Y$ are $N\times 1$ vectors of the assignments and outcomes, $a(\cdot)$ is a fixed function, and $b(\cdot)$ is an $N\times 1$ vector of fixed functions. The goal of the researcher is to pick the design and estimator to minimize worst-case MSE over the unknown potential outcomes. Formally, for $Y(\cdot)=(Y_i(0),Y_i(1))_{i=1}^N$, they solve:
\begin{align}
\inf _{(\delta,\hat\beta)}\sup_{Y(\cdot)\in [L,U]^{2N}}\E_\delta \left[(\hat\beta - \beta)^2\right].\label{eq:objective}
\end{align}
\noindent It's worth noting that while potential outcomes are treated here as fixed, nothing would change in the minimax analysis if $Y(\cdot)$ were instead jointly random, provided its joint distribution is otherwise unrestricted.\footnote{This is simply because, writing $\mathcal{P}$ as the set of joint distributions supported on $[L,U]^{2N}$, the inner problems are equivalent: $\sup_{P\in \mathcal{P}}\mathbb{E}_{P,\delta} \left[(\hat\beta - \beta)^2\right]=\sup_{Y(\cdot)\in[L,U]^{2N}}\E_\delta \left[(\hat\beta - \beta)^2\right]$.} By contrast, additional \emph{a priori} restrictions on the distribution of $Y(\cdot)$ could change the optimal design; with \emph{iid} sampling, for example, minimax experimental designs can become unnatural  \citep{borusyakHullMunro2026robust}. 

We first analyze this problem for the large set of affine estimators that satisfy a natural equivariance condition: 
\begin{align}
b(D)^\prime \left(D-\frac{1}{2}\mathbf{1}\right)=1,\hspace{0.3cm} \delta\text{-almost surely}\label{eq:equivariance}
\end{align}
\noindent To see why this condition is natural, note that we can rewrite the outcome equation as:
\begin{align*}
Y_i = \frac{Y_i(1)+Y_i(0)}{2}+(Y_i(1)-Y_i(0))\left(D_i-\frac{1}{2}\right).
\end{align*}
Thus, under equation \eqref{eq:equivariance}, shifting all treatment effects by some constant $t$ while keeping the location (i.e., midpoint) of potential outcomes fixed increases the estimate by exactly $t$: 
\begin{align*}
\hat\beta^{new} &= a(D) + \sum_ib(D)_i \left(\frac{Y_i(1)+Y_i(0)}{2}+(Y_i(1)-Y_i(0)+t)\left(D_i-\frac{1}{2}\right)\right) \\
&= a(D)+b(D)^\prime Y + b(D)^\prime\left(D-\frac{1}{2}\mathbf{1}\right)\times t = \hat\beta + t.
\end{align*} 
In this sense, equation \eqref{eq:equivariance} ensures the estimator is equivariant to midpoint-preserving treatment effect shifts. One can verify that it is satisfied for many common estimators---including simple or weighted difference-in-means, ordinary or weighted least squares with an intercept and possibly other predetermined controls or fixed effects, and matching estimators that match on predetermined strata---with the analogous equivariance holding for many nonlinear procedures such as differences of medians or trimmed means.\footnote{The general condition can be stated, for estimators $\hat\beta(D,Y)$, as $\hat\beta(D,Y+t(D-\frac{1}{2}\mathbf{1}))=\hat\beta(D,Y)+t$.}

The first result shows independent randomization is optimal for this class of affine estimators, with an unconventional estimator: 
\begin{proposition}\label{prop1}
Let $M=\frac{U+L}{2}$. Restricting to affine estimators satisfying equation \eqref{eq:equivariance}, 
\begin{align*}
\inf _{(\delta,\hat\beta)}\sup_{Y(\cdot)\in [L,U]^{2N}}\E_\delta \left[(\hat\beta - \beta)^2\right] = \frac{(U-L)^2}{N}\equiv V^*.
\end{align*}
A minimax procedure is independent random assignment, $D_i\stackrel{iid}{\sim}Bernoulli(0.5)$, with the unbiased estimator
\begin{align}
\hat \beta^*  = \frac{\sum_i (D_i-\frac{1}{2})(Y_i-M)}{\sum_i (D_i-\frac{1}{2})^2} = \frac{2}{N}\sum_{i=1}^N (2D_i-1)(Y_i-M).\label{eq:betastar}
\end{align}
\end{proposition}
\noindent The appendix proof is straightforward. It shows that, for any design and any estimator satisfying \eqref{eq:equivariance}, the adversary can generate an MSE of at least $(U-L)^2/N$ by placing both potential outcomes at one of the two support endpoints: i.e., selecting a $B_i\in\{L,U\}$ for each unit $i$ and setting $Y_i(0)=Y_i(1)=B_i$. If treatment assignments are correlated across units, the adversary can choose the endpoint pattern to align with that dependence and increase MSE above the lower bound. Independent random assignment eliminates this possibility and attains the $V^*$ bound when paired with the optimal estimator $\hat\beta^*$.

The $\hat\beta^*$ estimator is nonstandard. Equation \eqref{eq:betastar} shows it arises from an ordinary least squares regression of the support-midpoint-centered outcome $Y_i-M$ on the recentered treatment $D_i-1/2$, with no intercept. Equivalently, it is the Horvitz–Thompson ATE estimator applied to the centered outcome.\footnote{The second expression in \eqref{eq:betastar} can also be written $\hat\beta^*=\frac{1}{N}\sum_i \left(\frac{D_i(Y_i-M)}{\pi}-\frac{(1-D_i)(Y_i-M)}{1-\pi}\right)$ for $\pi=1/2$.} While omitting an intercept from a regression or using unnormalized Horvitz-Thompson weights is often undesirable, here it is exactly what allows the estimator to attain the minimax bound under independent randomization, by avoiding the slight negative correlation in observation weights induced by demeaning. Recentering the treatment by its known propensity score of $1/2$ ensures unbiasedness, while centering the outcome at the support midpoint anchors outcomes, leaving only midpoint deviations that cannot be exploited by the adversary under independent randomization. 

A closely related result shows that complete randomization is strictly suboptimal:
\begin{proposition}\label{prop2}
Let $N$ be even and restrict to affine estimators satisfying equation \eqref{eq:equivariance}. Then if $\delta$ is chosen to be balanced complete randomization,
\begin{align*}
\inf _{\hat\beta}\sup_{Y(\cdot)\in [L,U]^{2N}}\E_\delta \left[(\hat\beta - \beta)^2\right] \ge \frac{(U-L)^2}{N-1} =\frac{N}{N-1} V^*.
\end{align*}
Moreover, the difference-in-means estimator attains this bound. 
\end{proposition}
\noindent The optimality of difference-in-means for complete randomization is unsurprising; the new insight is that this conventional estimator-design pair has strictly higher worst-case MSE than the $V^*$ bound. Intuitively, complete randomization generates a slight negative correlation between treatment assignment pairs which the adversary can exploit by making exactly half of the $B_i$ equal to $L$ and half equal to $U$.  Averaging over such endpoint patterns makes MSE proportional to the squared magnitude of the difference-in-means outcome weights $b(D)$, with the usual $N/(N-1)$ finite-population factor from sampling without replacement. 

The $N/(N-1)$ factor inflating worst-case MSE with complete randomization is, of course, small for even moderately large populations; repeated exact-balance restrictions, however, can produce a first-order loss. To illustrate, consider paired randomization with $J$ predetermined pairs and the conventional pair-fixed-effect regression estimator. Here there are $J$ exact-balance restrictions with exactly one unit selected for treatment in each pair, and assignments are perfectly negatively correlated within pairs. The adversary can exploit this restriction by choosing a no-effect schedule in which one unit’s potential outcomes equal $L$ and the other’s equal $U$ in every pair. Each pair then contributes an independent estimation error of magnitude $U-L$, yielding worst-case MSE that is twice that of $V^*$. 
\begin{proposition}\label{pair_prop}
Suppose $N = 2J$ units are partitioned into $J$ pairs, indexed by $(j, 1)$ and $(j, 2)$. Restrict to affine estimators satisfying \eqref{eq:equivariance} and set $\delta$ to paired randomization with exactly one unit in each pair assigned with probability $1/2$, independently across pairs. Then:
\begin{align*}
\inf_{\hat\beta}\sup_{Y(\cdot)\in [L,U]^{2N}}\E_\delta \left[(\hat\beta - \beta)^2\right] \ge \frac{(U-L)^2}{J} =2 V^*.
\end{align*}
Moreover, the pair-fixed-effect regression estimator attains this bound. 
\end{proposition}

A minimax analysis is likely too pessimistic for the expected MSE of such a paired randomization strategy, as it does not incorporate any \emph{a priori} information on the similarity of potential outcomes within pairs. In practice, stratification is typically based on such information and can significantly improve MSE (see, e.g., Bai 2022; Tabord-Meehan 2023). In such settings, Proposition \ref{pair_prop} can be viewed as a formalization of the idea that stratification can come at a high MSE cost when the strata are not overly informative about potential outcomes---such that highly heterogeneous potential outcomes could coexist within strata.

\section{General Problem}\label{sec:general}

We now consider the minimax problem \eqref{eq:objective} without restricting the class of estimators. The main result shows that independent randomization remains optimal while the optimal estimator is generally nonlinear:
\begin{theorem}\label{theorem1}
For nonnegative integers $k,m$ with $k+m\leq N$, define $\theta_{k,m}=\frac{2k+m}{N}$ and 
\begin{align*}
p_{k,m}(x)=2^{-m}\binom{m}{x-k},
\quad x=0,\ldots,N,
\end{align*}
where $\binom{m}{j}=0$ when $j\notin\{0,\ldots,m\}$. Let $\cK_N$ be the set of all such $(k,m)$ and define
\begin{align}\label{eq:problem}
\kappa_N
=
\min_{d=(d_0,\ldots,d_N)\in[0,2]^{N+1}}
\max_{(k,m)\in\cK_N}
\sum_{x=0}^N
p_{k,m}(x)(d_x-\theta_{k,m})^2.
\end{align}
Then, over all measurable estimators,
\begin{align*}
\inf _{(\delta,\hat\beta)}\sup_{Y(\cdot)\in [L,U]^{2N}}\E_\delta \left[(\hat\beta - \beta)^2\right] =  (U-L)^2 \kappa_N.
\end{align*}
A minimax procedure is independent random assignment, with the estimator
\begin{align*}
\hat\beta^*_{NL} = \hat\beta^* + (U-L)\sum_{x=0}^N\left(d_x^*-\frac{2x}{N}\right)p_x(Z)
\end{align*}
where $\hat\beta^*$ is the optimal affine estimator in Proposition \ref{prop1},  $d^*=(d_0^*,\dots,d_N^*)$ is any solution to \eqref{eq:problem}, $p_x(z)=Pr(\sum_i B_i=x)$ for $B_i\stackrel{ind}{\sim}Bernoulli(z_i)$, and $Z=(Z_i)_{i=1}^N$ for
\begin{align*}
Z_i=\frac{1}{2}+\frac{(2D_i - 1)(Y_i-\frac{U+L}{2})}{U-L}\in[0,1].
\end{align*}
\end{theorem}
\noindent The appendix proof follows in the same spirit of \cite{hodges1982minimax}; as in their bounded finite-population sampling problem, the minimax estimator $\hat\beta^*_{NL}$ is shown to be a Bayes rule under a least-favorable prior supported on the potential outcome endpoints.\footnote{The minimax characterization in Theorem \ref{theorem1} was independently obtained by \citet[][Theorems 2.1-2.3]{sudijono2026sharp}: their $(p,q,r)$ corresponds to my $(k,N-k-m,m)$, and their optimal endpoint rule satisfies $f_N^*(x)=d_x^*-1$.} This estimator takes the form of the previous restricted-affine $\hat\beta^*$ plus a term that is nonlinear in $Y$. This added term generally makes $\hat\beta^*_{NL}$ biased while reducing worst-case MSE.

To understand this new estimator, it is useful to specialize it to the previous affine class while relaxing the equivariance restriction \eqref{eq:equivariance}. The next result shows that a ``shrunk'' version of $\hat\beta^*$ is then optimal (along with independent randomization):
\begin{corollary}\label{corollary1}
Over all affine estimators, without imposing  \eqref{eq:equivariance}, minimax MSE is $(U-L)^2/(\sqrt{N}+1)^2$. A minimax procedure is independent random assignment with the estimator
\begin{align*}
\hat\beta^*_{AF} = \frac{\sqrt{N}}{\sqrt{N}+1}\hat\beta^*
\end{align*}
Moreover, $\kappa_N<(\sqrt{N}+1)^{-2}$ when $N>2$ such that $\hat\beta^*_{NL}$ has strictly lower worst-case MSE.
\end{corollary}
\noindent The first part of this corollary shows that shrinking $\hat\beta^*$ towards zero---generating some bias and violating the equivariance restriction---reduces worst-case MSE by a factor of $N/(\sqrt{N}+1)^2$. The second part shows that further expanding to nonlinear estimators is useful for any nontrivial population size $N>2$. 

Computing $\hat\beta^*_{NL}$ is relatively straightforward, as $\eqref{eq:problem}$ is a finite-dimensional convex program with quadratic constraints that depends only on the population size $N$. A direct solver is practical for populations in low thousands; for larger $N$, constraint generation gives certified lower and upper bounds. Numerically, I find a relatively small share of constraints bind.\footnote{E.g., the number of binding or nearly binding constraints is in the low tens when $N$ is in the hundreds.}

Figure \ref{fig:minimax_risk} shows the potential benefit in using $\hat\beta^*_{NL}$ instead of $\hat\beta^*$. Specifically, it plots $\kappa_N N$, the worst-case MSE from using $\hat\beta_{NL}$ as a fraction of the worst-case MSE from using $\hat\beta^*$ (both with independent randomization)  on a grid up to $N=2,500$.  The relative worst-case MSE reduction is substantial for moderate sample sizes and declines smoothly with $N$; for example, it is around 30\% when $N=100$ and falls to around 15\% for $N=1,000$.  Proposition \ref{propA1} in Appendix \ref{apdx_kappaN} shows it converges to zero as $N\rightarrow\infty$, as with the gain in using $\hat\beta^*$ and independent randomization relative to difference-in-means estimation. with complete randomization.\footnote{\cite{sudijono2026sharp} sharpen this result by showing $\kappa_N = N^{-1} - C_A N^{-4/3}+o(N^{-4/3})$, or $N\kappa_N = 1 - C_A N^{-1/3}+o(N^{-1/3})$, for constant $C_A$. One may wonder whether complete randomization remains strictly suboptimal with general estimators. I have verified this numerically in populations of moderate size and conjecture it is true in general.} Figure \ref{fig:minimax_risk} also plots the corresponding relative worst-case MSE of the best affine estimator $\hat\beta^*_{AF}$ under independent randomization, which has a similar shape.

\begin{figure}[t]
    \centering
    \caption{Relative Worst-Case MSE of Optimal Estimators}
   \vspace{0.2cm}
    \includegraphics[width=0.75\textwidth]{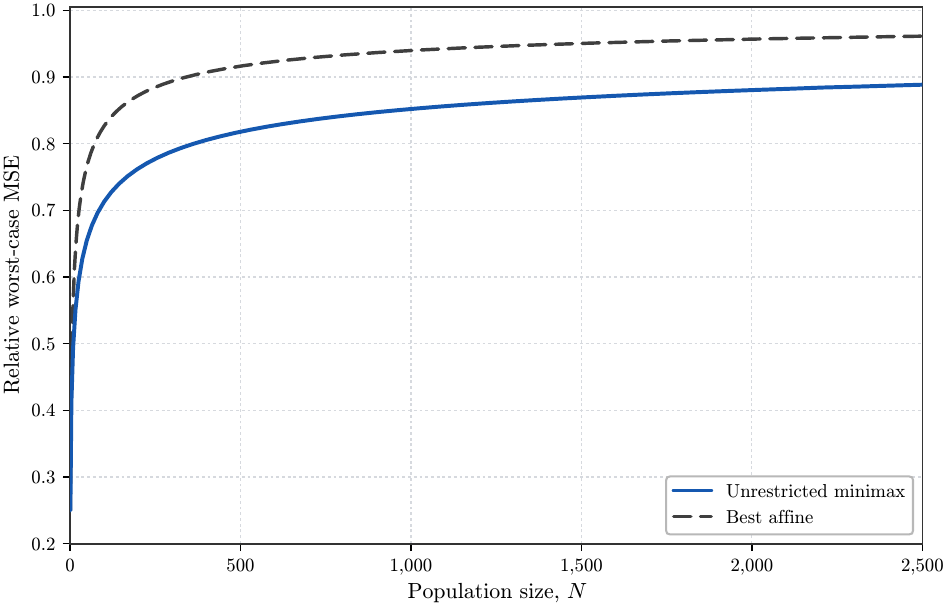}
    \label{fig:minimax_risk}
    
    \begin{minipage}{0.9\textwidth}
    \footnotesize
    \textit{Notes:} The solid blue line plots $N\kappa_N$ against the actual population size $N$, where $\kappa_N$ is the minimax MSE from Theorem \ref{theorem1} divided by $(U-L)^2$. Thus, it reports the worst-case MSE of the unrestricted minimax estimator as a fraction of the worst-case MSE of the equivariant affine
estimator $\widehat\beta^*$ from Proposition \ref{prop1} (both under independent
randomization). Values solve program \eqref{eq:problem} on the grid $N=j^2$, $j=1,\ldots,50$. The dashed black line plots the corresponding relative worst-case MSE of the best affine estimator under independent randomization, $N/(\sqrt{N}+1)^2$.
    \end{minipage}
\end{figure}
 
Figure \ref{fig:binarysimplex} shows how the unconventional estimators $\hat\beta^*_{NL}$ and $\hat\beta^*$ generate worst-case MSE reductions by reporting their actual MSE across different configurations of binary potential outcomes. Specifically, I consider a population of $N=1,000$ units with type shares $\pi_{ab}=\frac{|Y_i(0)=a,Y_i(1)=b|}{1,000}$ parameterized by $e=\pi_{00}+\pi_{11}$, $\ell=\pi_{11}-\pi_{00}$, and the ATE $\beta=\pi_{01}-\pi_{10}$. Here $e$ is the share of units with equal potential outcomes at the same support endpoint---types \(00\) and \(11\)---while \(\ell\) is the excess share of type \(11\) over type \(00\) capturing the ``imbalance'' of these endpoint units. The figure shows MSE for all configurations of potential outcomes that satisfy the feasibility constraint of $|\ell|\le e\le 1-|\beta|$, in a grid where each $\pi_{ab}$ is a multiple of $1/8$. Panel (a) shows this for $\beta=0$ while panel (b) uses $\beta=0.25$. Exact MSE is reported for the nonlinear minimax estimator $\hat\beta^*_{NL}$, the outcome-centered Horvitz-Thompson estimator $\hat\beta^*$ (both under independent randomization), and the difference-in-means estimator under complete randomization.\footnote{Here $\mathrm{MSE}=e/N$ for $\hat\beta^*$ while for difference-in-means  \(\mathrm{MSE}=(e-\ell^{2})/(N-1)\). For $\hat\beta^*_{NL}$, let \(n_{ab}=N\pi_{ab}\) and set \(k=n_{01}\), \(m=n_{00}+n_{11}\), and \(\theta=(2k+m)/N=1+\beta\). Then $\hat\beta^*_{NL}=d^{*}_{k+B}-1$ for  \(B\sim\mathrm{Binomial}(m,1/2)\), where \(d^{*}\) solves \eqref{eq:problem}, giving $\mathrm{MSE}=\sum_{b=0}^{m}\binom{m}{b}2^{-m}\left(d^{*}_{k+b}-\theta\right)^2$ . Plots look similar for other choices of $N$.\label{fn_MSE}}

\begin{figure}[h!]
    \centering
    \caption{Pointwise MSE under Binary Potential Outcomes}
   \vspace{0.1cm}
    \includegraphics[width=0.86\textwidth]{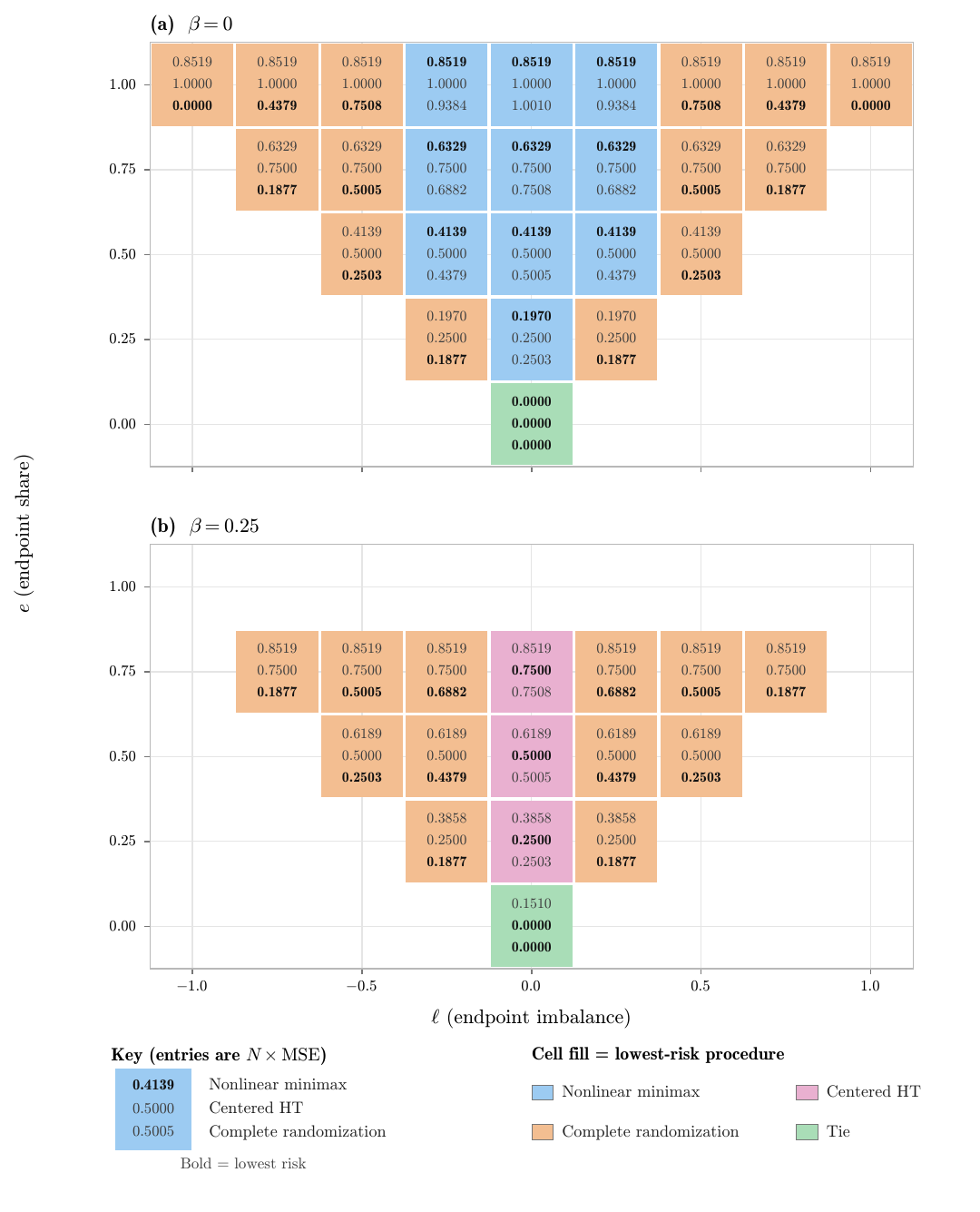}
    \label{fig:binarysimplex}
    
    \begin{minipage}{1\textwidth}
    \footnotesize
    \textit{Notes:} This figure reports the pointwise MSE of different designs and estimators across binary potential outcome schedules with \(N=1{,}000\). As detailed in the text, $e$ is the share of units with both potential outcomes at a support endpoint and $\ell$ captures the imbalance of these endpoint units. Panels (a) and (b) consider schedules with an ATE of zero and 0.25, respectively. Within each cell, the entries from top to bottom report \(N\) times MSE for the unrestricted nonlinear minimax estimator \(\hat{\beta}^{*}_{NL}\) from Theorem 1 under independent randomization, the support-midpoint-centered Horvitz–Thompson estimator \(\hat{\beta}^{*}\) from Proposition 1 under independent randomization, and difference-in-means under balanced complete randomization. The nonlinear rule is obtained by solving program \eqref{eq:problem}. All displayed risks are then evaluated exactly over the assignment distribution, conditional on the indicated finite population. Bold denotes the lowest MSE within a cell, and the cell fill identifies the corresponding procedure; green denotes a tie.
    \end{minipage}
\end{figure}

Independent randomization can beat complete randomization when there are some units with no treatment effects and both potential outcomes at the bounds (here, $Y_i(0)=Y_i(1)=1$ or $Y_i(0)=Y_i(1)=0$), especially when these endpoint units are roughly balanced (i.e., roughly evenly split between zero and one, with $|\ell|\approx 0$). Intuitively, the midpoint centering in $\hat\beta^*$ is then more effective at reducing variance, which $\hat\beta^*_{NL}$ further improves on. It does so by introducing some bias, essentially shrinking the estimate towards zero. When $\beta=0$ this in fact adds no bias, so $\hat\beta^*_{NL}$ weakly beats $\hat\beta^*$ everywhere (and strictly so for $e>0$) while also having more scope to improve MSE overall; in contrast, when $\beta=0.25$, the unbiased $\hat\beta^*$ beats the biased $\hat\beta^*_{NL}$ while being optimal over a smaller share of configurations. 

Researchers can use calculations like those in Figure \ref{fig:binarysimplex}  to compute the expected MSE (or, more generally, bias and variance) of different designs and ATE estimators for binary outcomes.  Writing $\mu_1$ as the population-average treated potential outcome and $\mu_0$ as the population-average untreated potential outcome, we have $\beta=\mu_1-\mu_0$, $\ell=\mu_1+\mu_0-1$, and $e\in[|\ell|,1-|\beta|]$.  Thus, for any predicted outcome means in the treatment and control arm (and anticipated population size), the researcher can solve \eqref{eq:problem} for the implied values of $\beta$, $\ell$, and range of $e$ and compute the expected MSE of different designs and estimators (see again footnote \ref{fn_MSE}). For example, the areas of Figure \ref{fig:binarysimplex}(a) where $\hat\beta^*_{NL}$ dominate are those where $\mu_0=\mu_1\in\{3/8,1/2,5/8\}$, while the areas of Figure \ref{fig:binarysimplex}(b) where $\hat\beta^*$ dominate are those where $\mu_0=3/8$ and $\mu_1=5/8$. Similar computations follow for other (nonbinary) outcome bounds.

\section{Conclusion}\label{sec:conclusion}

This paper shows a curious advantage of independent randomization, relative to complete or otherwise stratified randomization, when estimating ATEs for bounded potential outcomes. Exact balance of treated shares can be suboptimal once the level of outcomes is anchored, because the estimator variance can no longer be made arbitrarily large by a random treated share multiplying some arbitrarily high or low outcomes. Independent randomization is generally preferred by avoiding negative correlations in treatment assignments which can otherwise blow up MSE. While such correlations tend to be small for complete randomization, they can be first-order if balance is imposed repeatedly, such as in paired randomization.

In practice, complete, paired, or more elaborate randomization schemes are likely to dominate simple independent randomization when based on informative observables. Similarly, more elaborate regressions that incorporate informative controls or fixed effects are likely to dominate the unconventional support-midpoint-centered estimator in the restricted affine class and likely also the minimax-optimal estimator $\hat\beta^*_{NL}$.  Extending this paper's minimax problem to additional assumptions on the homogeneity of potential outcomes within strata, or other prognostic structure, is a natural direction for future work.

\singlespacing
\bibliographystyle{aer}
\bibliography{causalHL}

\appendix 
\onehalfspacing
\numberwithin{proposition}{section}
\section{Proofs}
\subsection{Proposition \ref{prop1}}
\begin{proof}
Write $M=\frac{U+L}{2}$, $r=\frac{U-L}{2}$, and $X(D)=D-\frac{1}{2}\mathbf{1}$.  For each fixed $s\in\{-1,1\}^N$, consider the no-effect potential outcome schedule $Y_i(0)=Y_i(1)=M+rs_i$.  For any design $\delta$, the maximum risk is at least the average risk over these $2^N$ schedules.  Hence, for any affine estimator $\widehat\beta=a(D)+b(D)'Y$ satisfying  \eqref{eq:equivariance},
\begin{align*}
\sup_{Y(\cdot)\in[L,U]^{2N}}\E_\delta[(\widehat\beta-\beta)^2]
&\geq \E_\delta\left[2^{-N}\sum_s\{a(D)+Mb(D)'\1+r b(D)'s\}^2\right]\\
&=\E_\delta\left[\{a(D)+Mb(D)'\1\}^2+r^2\|b(D)\|^2\right].
\end{align*}
The second line uses $2^{-N}\sum_s s=0$ and $2^{-N}\sum_s ss'=I$.  By  \eqref{eq:equivariance}  and Cauchy--Schwarz,
\[
1=\{b(D)'X(D)\}^2\leq \|b(D)\|^2\|X(D)\|^2
=\frac N4\|b(D)\|^2,
\]
so $\|b(D)\|^2\geq 4/N$ a.s. Thus every admissible procedure has worst-case MSE of at least
\[
\frac{4r^2}{N}=\frac{(U-L)^2}{N}.
\]

It remains to show attainability.  Under independent Bernoulli$(0.5)$ assignment write $H_i=2D_i-1$, $\mu_i=\frac{Y_i(1)+Y_i(0)}2$, and $\tau_i=Y_i(1)-Y_i(0)$. Then $Y_i-M=(\mu_i-M)+H_i\tau_i/2$, so the proposed estimator satisfies
\[
\widehat\beta^*=\frac2N\sum_iH_i(Y_i-M)
=\beta+\frac2N\sum_iH_i(\mu_i-M).
\]
Then, since the $H_i$ are independent mean-zero Rademacher variables,
\[
\E_\delta[(\widehat\beta^*-\beta)^2]
=\frac4{N^2}\sum_i(\mu_i-M)^2
\leq \frac{4r^2}{N}
=\frac{(U-L)^2}{N}.
\]
The bound is attained, for example, when $Y_i(0)=Y_i(1)=U$ for every $i$.  
\end{proof}
\subsection{Proposition \ref{prop2}}

\begin{proof}
Write $M=\frac{U+L}{2}$, $r=\frac{U-L}{2}$,  $X(D)=D-\frac{1}{2}\mathbf{1}$, $P=I-\frac1N\1\1'$, and $\mathcal S=\{s\in\{-1,1\}^N:\ \1's=0\}$. Under balanced complete randomization, $X(D)\in\1^\perp$ and $\|X(D)\|^2=N/4$ a.s.  For each $s\in\mathcal S$, consider the no-effect potential outcome schedule $Y_i(0)=Y_i(1)=M+rs_i$.  Uniform averaging over $\mathcal S$ gives
\[
\frac1{|\mathcal S|}\sum_{s\in\mathcal S}s=0 \hspace{0.3cm}\text{ and }\hspace{0.3cm}
\frac1{|\mathcal S|}\sum_{s\in\mathcal S}ss'=\frac{N}{N-1}P.
\]
Hence, for any affine estimator satisfying  \eqref{eq:equivariance} ,
\begin{align*}
\sup_{Y(\cdot)\in[L,U]^{2N}}\E_\delta[(\widehat\beta-\beta)^2]
&\geq \E_\delta\left[\{a(D)+Mb(D)'\1\}^2
+r^2\frac{N}{N-1}b(D)'Pb(D)\right].
\end{align*}
Since $X(D)\in\1^\perp$, $1=b(D)'X(D)=\{Pb(D)\}'X(D)$; hence $b(D)'Pb(D)\geq 4/N$ by Cauchy--Schwarz.  This yields the lower bound
\[
\sup_{Y(\cdot)\in[L,U]^{2N}}\E_\delta[(\widehat\beta-\beta)^2]
\geq \frac{(U-L)^2}{N-1}.
\]

For the upper bound, let $H=2D-\1$ and write $\mu_i=\{Y_i(1)+Y_i(0)\}/2$.  Under balanced complete randomization, the difference-in-means estimator can be written
\[
\widehat\beta_{\rm DM}=\frac2N H'Y
=\beta+\frac2N H'\mu
=\beta+\frac2N H'(\mu-M\1),
\]
because $H'\1=0$.  Moreover, $\E_\delta[HH']=\frac{N}{N-1}P$.  Thus
\begin{align*}
\E_\delta[(\widehat\beta_{\rm DM}-\beta)^2]
&=\frac4{N^2}(\mu-M\1)'\frac{N}{N-1}P(\mu-M\1)\\
&\leq \frac4{N(N-1)}\|\mu-M\1\|^2 \\
&\leq \frac{(U-L)^2}{N-1}.
\end{align*}
Equality is attained by a no-effect schedule with $N/2$ units having $Y_i(0)=Y_i(1)=U$ and $N/2$ units having $Y_i(0)=Y_i(1)=L$. Hence difference in means attains the lower bound.
\end{proof}

\subsection{Proposition \ref{pair_prop}}

\begin{proof}
Again let $M=\frac{U+L}{2}$ and $r=\frac{U-L}{2}$.  For $s=(s_1,\ldots,s_J)\in\{-1,1\}^J$, consider the no-effect potential outcome schedule:
\[
Y_{j1}(0)=Y_{j1}(1)=M+rs_j,
\hspace{0.3cm}\text{ and }\hspace{0.3cm}
Y_{j2}(0)=Y_{j2}(1)=M-rs_j.
\]
Averaging risk over these $2^J$ schedules gives, for any affine estimator satisfying  \eqref{eq:equivariance},
\begin{align*}
\sup_{Y(\cdot)\in[L,U]^{2J}}\E_\delta[(\widehat\beta-\beta)^2]
&\geq \E_\delta\left[\{a(D)+Mb(D)'\1\}^2
+r^2\sum_{j=1}^J\{b_{j1}(D)-b_{j2}(D)\}^2\right].
\end{align*}
Let $H_j=2D_{j1}-1$, so $2D_{j2}-1=-H_j$.  Condition \eqref{eq:equivariance} implies, pointwise in $D$,
\[
1=b(D)'\left(D-\frac12\1\right)
=\frac12\sum_{j=1}^JH_j\{b_{j1}(D)-b_{j2}(D)\}.
\]
By Cauchy--Schwarz,
\[
\sum_{j=1}^J\{b_{j1}(D)-b_{j2}(D)\}^2\geq \frac4J.
\]
Therefore every admissible estimator has worst-case risk at least
\[
\frac{4r^2}{J}=\frac{(U-L)^2}{J}.
\]

Under paired randomization, the pair-fixed-effect coefficient can be written
\[
\widehat\beta_{\rm PFE}
=\frac1J\sum_{j=1}^JH_j(Y_{j1}-Y_{j2}).
\]
Writing $\mu_{jr}=\{Y_{jr}(1)+Y_{jr}(0)\}/2$, direct substitution gives
\[
\widehat\beta_{\rm PFE}-\beta
=\frac1J\sum_{j=1}^JH_j(\mu_{j1}-\mu_{j2}).
\]
The $H_j$ are independent mean-zero Rademacher variables, so
\[
\E_\delta[(\widehat\beta_{\rm PFE}-\beta)^2]
=\frac1{J^2}\sum_{j=1}^J(\mu_{j1}-\mu_{j2})^2
\leq \frac{(U-L)^2}{J}.
\]
Equality is attained by setting, in each pair, one unit's two potential outcomes equal to $U$ and the other's equal to $L$.  Thus the pair-fixed-effect estimator attains the lower bound.
\end{proof}
\subsection{Theorem \ref{theorem1}}
\begin{proof}
Write $M=\frac{U+L}{2}$. $\Delta=U-L$, $A_i=\frac{Y_i(1)-L}{\Delta}$, and $C_i=\frac{U-Y_i(0)}{\Delta}$. Then $A_i,C_i\in[0,1]$ and $\beta=\Delta(\theta-1)$ for
\[
\theta=\frac1N\sum_{i=1}^N(A_i+C_i).
\]
Moreover, 
\[
Z_i
=\frac12+\frac{(2D_i-1)(Y_i-M)}{\Delta}=D_iA_i+(1-D_i)C_i.
\]
It thus suffices to solve the normalized problem of estimating $\theta$ from one observed member of each pair $(A_i,C_i)$; squared-error risk for $\beta$ is $\Delta^2$ times the corresponding risk for $\theta$. 

We proceed in three steps:

\smallskip
\noindent\emph{Step 1: Symmetrization.}
Since $Y_i=M+\Delta(2D_i-1)(Z_i-\frac12)$, any measurable estimator based on $(D,Y)$ can be equivalently written as some $t(D,Z)$. Consider any design $\delta$ and any such estimator.  For $h\in\{0,1\}^N$, let $(A^h,C^h)$ be obtained by interchanging
$A_i$ and $C_i$ whenever $h_i=1$, and define $\delta_h(d)=\delta(d\oplus h)$ and $t_h(d,z)=t(d\oplus h,z)$. Since simultaneously replacing $(A,C)$ by $(A^h,C^h)$ and $D$ by
$D\oplus h$ leaves both $Z$ and $\theta$ unchanged,
\[
R(\delta_h,t_h;A,C)
=
R(\delta,t;A^h,C^h).
\]
The parameter space is invariant under such swaps, so every
$(\delta_h,t_h)$ has the same maximum risk as $(\delta,t)$. Mixing uniformly over all $h$ therefore cannot increase maximum risk. The resulting design satisfies
\[
Pr(D=d)=2^{-N}\sum_h\delta(d\oplus h)=2^{-N}.
\]
so the assignment becomes independent Bernoulli$(0.5)$ randomization.  Conditional on $(D,Z)=(d,z)$, the mean of the randomized estimator is
\[
\sum_h\delta(d\oplus h)t(d\oplus h,z)
=
\sum_r\delta(r)t(r,z)
\equiv \widetilde t(z),
\]
which does not depend on $d$. By Jensen's inequality, replacing the
randomized estimate by this conditional mean weakly lowers squared-error
risk.  Thus it is without loss for minimaxity to restrict attention to independent randomization and estimators of the form $\widetilde t(Z)$.  Since independent randomization is invariant to permutations and the
parameter space is invariant to a common relabeling of the potential outcome pairs, the same averaging-and-Jensen argument over permutations allows $\widetilde t$ to be taken symmetric.

\smallskip
\noindent\emph{Step 2: Lower bound.}
Restrict the transformed parameter space to vertices $(A_i,C_i)\in\{0,1\}^2$ for every $i$.  Up to pair swaps and permutations, a vertex is characterized by $k$ pairs of type $(1,1)$, $m$ mixed pairs, and $N-k-m$ pairs of type $(0,0)$.  Its target is $\theta_{k,m}=(2k+m)/N$.  Under independent randomization, if $X=\sum_iZ_i$, then $X\overset d= k+\operatorname{Binomial}(m,1/2)$ so $Pr(X=x)=p_{k,m}(x)$.  On binary observations a symmetric estimator is completely described by numbers $d_x=\widetilde{t}(z)$ whenever $\sum_i z_i=x$.  Since $\theta\in[0,2]$, clipping the $d_x$ to $[0,2]$ cannot increase risk.  Since the symmetrized procedure has maximum risk no greater than the original procedure, while its maximum risk is at least maximum risk
over the endpoint parameter space, it follows that every original
procedure has maximum risk of at least
\[
\min_{d\in[0,2]^{N+1}}\max_{(k,m)\in K_N}
\sum_{x=0}^Np_{k,m}(x)(d_x-\theta_{k,m})^2
=\kappa_N.
\]

\smallskip
\noindent\emph{Step 3: Attainment.}
For $d\in[0,2]^{N+1}$ define 
\[
\mathcal B_d(z)=\sum_{x=0}^Nd_xp_x(z),
\]
where $p_x(z)$ is the probability that a sum of independent Bernoulli$(z_i)$ variables equals $x$.  The function $\mathcal B_d$ is affine in each coordinate separately and equals $d_{\sum_i z_i}$ at every $z\in\{0,1\}^N$.

Under independent randomization, the normalized risk at $(A,C)\in[0,1]^{2N}$ is
\[
R_d(A,C)=2^{-N}\sum_{g\in\{0,1\}^N}
\left[\mathcal B_d\{Z_g(A,C)\}-\frac1N\sum_i(A_i+C_i)\right]^2.
\]
Fix all coordinates except $A_i$.  In terms with $g_i=1$, both the estimate and the target are affine in $A_i$; in terms with $g_i=0$, the estimate is constant and the target is affine.  Each squared difference is therefore convex in $A_i$.  The same argument applies separately to every $A_i$ and $C_i$.  Thus $R_d$ is separately convex on the cube, and successively maximizing over coordinates shows that its maximum is attained at a vertex.  Taking $d=d^*$ solving \eqref{eq:problem}, every vertex risk is at most $\kappa_N$ by construction, so
\[
\sup_{(A,C)\in[0,1]^{2N}}R_{d^*}(A,C)\leq\kappa_N.
\]
Along with Step 2, this shows the normalized minimax value of $\kappa_N$ and hence the stated value $\Delta^2\kappa_N$. Finally, note that
\[
\sum_{x=0}^N xp_x(Z)=\sum_{i=1}^NZ_i,
\hspace{0.3cm}\text{ and }\hspace{0.3cm}
\widehat\beta^*=\Delta\left(\frac2N\sum_iZ_i-1\right).
\]
Therefore we can write:
\[
\hat\beta^*_{NL}=\Delta\{\mathcal B_{d^*}(Z)-1\}
=\widehat\beta^*+\Delta\sum_{x=0}^N\left(d_x^*-\frac{2x}{N}\right)p_x(Z).
\]

\smallskip
\noindent\emph{Least-favorable-prior characterization.}
For completeness in connecting to \cite{hodges1982minimax}, let $\pi$ be any prior on the finite state space $K_N$.  Because $[0,2]^{N+1}$ and the simplex of priors on $\mathcal K_N$
are compact and convex, while the prior-weighted risk is continuous,
convex in $d$, and linear in $\pi$, the minimax theorem gives
\[
\kappa_N
=
\max_\pi\min_d
\sum_{k,m}\pi_{k,m}
\sum_x p_{k,m}(x)(d_x-\theta_{k,m})^2.
\]
Let $d^*$ be a minimax rule and $\pi^*$ a least-favorable prior.
They form a saddle point. Hence $d^*$ minimizes Bayes risk under
$\pi^*$ and, whenever
\[
q_{\pi^*}(x)
=
\sum_{k,m}\pi^*_{k,m}p_{k,m}(x)>0,
\]
\[
d_x^*
=
\frac{\sum_{k,m}\pi^*_{k,m}p_{k,m}(x)\theta_{k,m}}
     {q_{\pi^*}(x)}
=
\E_{\pi^*}[\theta_{k,m}\mid X=x].
\]
Coordinates with $q_{\pi^*}(x)=0$ are unrestricted by Bayes risk and
may be chosen as in the minimax rule $d^*$. 
\end{proof}

\subsection{Corollary \ref{corollary1}}
\begin{proof}
Work with the normalized problem for estimating $\theta=\frac{1}{N}\sum_{i=1}^N(A_i+C_i)$
as in the proof of Theorem \ref{theorem1}.  Since, for fixed $D=d$,
\[
Y_i=M+(U-L)(2d_i-1)(Z_i-\tfrac12),
\]
any estimator that is affine in $Y$ is also affine in $Z$ conditional on
$D$.  The conditional-mean symmetrization in the proof of Theorem 1 is a
weighted average $\sum_d\delta(d)t(d,Z)$ with weights independent of $Z$,
so it remains affine in $Z$. Permutation averaging preserves affinity
as well. Hence it is without loss to consider a symmetric affine
function of $Z$. Such an
estimator must take the form
\[
\widehat\theta
=
1+c+\lambda\left(\frac{2}{N}\sum_{i=1}^N Z_i-1\right)
\]
for some $c,\lambda\in\mathbb R$.  Equivalently, its endpoint rule in
\eqref{eq:problem} is $d_x=1+c+\lambda q_x$ and $q_x=\frac{2x}{N}-1$. 

For any $(k,m)$, let
\[
q_{k,m}=\theta_{k,m}-1=\frac{2k+m-N}{N}.
\]
Since $X=k+\operatorname{Binomial}(m,0.5)$, we have $E[q_X]=q_{k,m}$ and $\operatorname{Var}(q_X)=\frac{m}{N^2}.$
The normalized risk of the affine rule is therefore
\[
R_{k,m}(c,\lambda)
=
\frac{\lambda^2m}{N^2}
+
\left\{c+(\lambda-1)q_{k,m}\right\}^2.
\tag{A.1}\label{A1}
\]

Consider the three particular $(k,m)$ of $(0,N)$, $(0,0)$, and $(N,0)$. Equation \eqref{A1} gives
\[
R_{0,N}(c,\lambda)
=
c^2+\frac{\lambda^2}{N},
\]
while $\max\{R_{0,0}(c,\lambda),R_{N,0}(c,\lambda)\}
=
\max\{(c+1-\lambda)^2,(c-1+\lambda)^2\}
\geq (1-\lambda)^2$. Thus every affine procedure has normalized worst-case risk at least
\[
\max\left\{\frac{\lambda^2}{N},(1-\lambda)^2\right\}.
\]
Minimizing this expression over $\lambda$ gives $\lambda_N=\frac{\sqrt N}{\sqrt N+1}$ and 
\[
\min_\lambda
\max\left\{\frac{\lambda^2}{N},(1-\lambda)^2\right\}
=
\frac{1}{(\sqrt N+1)^2}.
\]

This lower bound is attained with $c=0$ and $\lambda=\lambda_N$.
Indeed, for any $(k,m)$, $|q_{k,m}|\leq 1-\frac{m}{N}$, and, since $\frac{\lambda_N^2}{N}=(1-\lambda_N)^2$, equation \eqref{A1} yields
\[
R_{k,m}(0,\lambda_N)
=
\frac{1}{(\sqrt N+1)^2}
\left\{
\frac{m}{N}+q_{k,m}^2
\right\}
\leq
\frac{1}{(\sqrt N+1)^2}
\left\{
\frac{m}{N}+
\left(1-\frac{m}{N}\right)^2
\right\}
\leq
\frac{1}{(\sqrt N+1)^2}.
\]

Let
\[
d_x^A=1+\lambda_N\left(\frac{2x}{N}-1\right).
\]
The preceding calculation shows that every endpoint risk of $d^A$ is
at most $(\sqrt N+1)^{-2}$. Its Bernstein extension is
\[
B_{d^A}(Z)
=
1+\lambda_N\left(\frac2N\sum_i Z_i-1\right),
\]
since $\sum_x x p_x(Z)=\sum_iZ_i$. Thus this extension is precisely
the normalized version of $\lambda_N\widehat\beta^*$. By the
separate-convexity argument in Step 3 of the proof of Theorem 1, its
maximum risk over the full cube is attained at a vertex. Hence the affine minimax risk is
\[
\frac{(U-L)^2}{(\sqrt N+1)^2}.
\]

Finally, the Bernstein extension of $d_x^A$ is itself affine, since
\[
\sum_{x=0}^N xp_x(Z)=\sum_{i=1}^N Z_i.
\]
Thus the corresponding estimator of $\beta$ is
\[
(U-L)\left\{
\sum_xd_x^Ap_x(Z)-1
\right\}
=
\lambda_N(U-L)
\left(\frac{2}{N}\sum_iZ_i-1\right)
=
\lambda_N\widehat\beta^*,
\]

For the strict inequality, let $r_N=\frac{1}{(\sqrt N+1)^2}$.  At the affine optimum,
\[
R_{k,m}(d^A)
=
r_N\left\{\frac{m}{N}+q_{k,m}^2\right\}.
\tag{A.2}\label{A2}
\]
The inequality above is strict except when $(k,m)$ equals $(0,N)$, $(0,0)$, or $(N,0)$. Indeed, writing $t=m/N$, we have
\[
\frac{m}{N}+q_{k,m}^2
\leq t+(1-t)^2\leq1.
\]
Equality in the second inequality requires $t\in\{0,1\}$.  If $t=1$,
then necessarily $(k,m)=(0,N)$.  If $t=0$, equality in the first
inequality additionally requires $|q_{k,0}|=1$, which gives
$(k,m)=(0,0)$ or $(N,0)$.  Thus these are exactly the  named values of $(k,m)$.

For $N > 2$, choose
\[
C>
\frac{2^{1-N}}{N^{-1}-2^{1-N}}
\]
and define a perturbation $h=(h_0,\ldots,h_N)$ by $h_0=-1$, $h_N=1$, and $h_x=-Cq_x,$ $x=1,\ldots,N-1$. The directional derivatives of risks at the first two active states
in \eqref{A2} are strictly negative:
\begin{align*}
\left.\frac{d}{d\varepsilon} R_{0,0}(d^A+\varepsilon h)\right|_{\varepsilon=0} = \left.\frac{d}{d\varepsilon} R_{N,0}(d^A+\varepsilon h)\right|_0 = -2(1-\lambda_N)<0.
\end{align*}
 At the third, if
$X\sim\operatorname{Binomial}(N,\tfrac12)$,
\[
\left.\frac{d}{d\varepsilon}
R_{0,N}(d^A+\varepsilon h)\right|_{\varepsilon=0}
=
2\lambda_N\mathbb E[q_Xh_X].
\]
Since $\mathbb E[q_X^2]=1/N$,
\[
\mathbb E[q_Xh_X]
=
2^{1-N}
-
C\left(\frac1N-2^{1-N}\right)
<0
\]
by the choice of $C$.  Thus all three active risks decrease for
sufficiently small $\varepsilon>0$.  Every other $(k,m)$ has
strict slack in \eqref{A2}, and there are only finitely many such states.
Hence, for sufficiently small $\varepsilon>0$, the perturbed rule
$d^A+\varepsilon h$ remains in $[0,2]^{N+1}$ and has every endpoint
risk strictly below $r_N$.  Thus by \eqref{eq:problem}, $\kappa_N<\frac{1}{(\sqrt N+1)^2}$ for every $N>2$.
\end{proof}

\subsection{Asymptotics of $N\kappa_N$ }\label{apdx_kappaN}

\begin{proposition}
$N\kappa_N\rightarrow 1$  as $N\rightarrow\infty$.\label{propA1}
\end{proposition}

\begin{proof}
For $(k,m)\in\mathcal K_N$, write $X=k+\operatorname{Binomial}(m,1/2)$, $ \theta_{k,m}=\frac{2k+m}{N}$, and 
\[
    R_{k,m}(d)
    =\mathbb E\bigl[(d_X-\theta_{k,m})^2\bigr].
\]
For the upper bound, take $d_x=2x/N$. This is feasible and unbiased
since $\mathbb E[X]=k+m/2$. Consequently,
\[
    R_{k,m}(d)
    =\frac{4}{N^2}\operatorname{Var}(X)
    =\frac{m}{N^2}
    \leq \frac{1}{N}.
\]
It follows that
\[
    \limsup_{N\to\infty}N\kappa_N\leq 1.
\]

For the lower bound, let $s=\lfloor N^{3/4}\rfloor$ and $n=N-s$. Restrict the maximum in equation \eqref{eq:problem} to the feasible states
$(k,m)=(a,n)$, $a=0,\ldots,s$. Put the uniform prior on these states.
Equivalently, let $A\sim\operatorname{Unif}\{0,\ldots,s\}$ and $ B\sim\operatorname{Binomial}(n,1/2)$ be independent, and observe $X=A+B$. The target under this prior is
\[
    \theta_A=\frac{n+2A}{N}.
\]
Since maximum risk is at least average risk,
\begin{align}
    \kappa_N
    &\geq
    \inf_{d\in[0,2]^{N+1}}
    \mathbb E\bigl[(d_X-\theta_A)^2\bigr] \notag\\
    &=\mathbb E\bigl[\operatorname{Var}(\theta_A\mid X)\bigr]
      =\frac{4}{N^2}
       \mathbb E\bigl[\operatorname{Var}(B\mid X)\bigr].
    \tag{A.3}\label{A3}
\end{align}
The first equality in the second line is the usual Bayes-risk identity; the
posterior mean is feasible because $\theta_A\in[0,2]$. The final equality
follows because $A=X-B$ conditional on $X$.

We are finished upon showing that
\[
    \mathbb E\bigl[\operatorname{Var}(B\mid X)\bigr]
    =(1-o(1))\frac{n}{4},
    \tag{A.4}\label{A4}
\]
since then $\kappa_N\geq (1-o(1))\frac{n}{N^2}$ by \eqref{A3} and, since $n/N\to1$,
\[
    \liminf_{N\to\infty}N\kappa_N\geq1.
\]

To show this, write $\mu=\frac{n}{2}$, $ t=\left\lceil\sqrt{2n\log n}\right\rceil$, and $\varepsilon_n=\Pr(|B-\mu|>t)$. Hoeffding's inequality gives $\varepsilon_n\leq 2n^{-4}$. Let $J=\{x\in\mathbb{Z}:\mu+t\leq x\leq\mu+s-t\}$. Since $A$ is uniform and independent of $B$ we have, for $x\in\text{supp}(X)$,
\begin{align*}
Pr(B=b\mid X=x)&=\frac{Pr(B=b)Pr(A=x-b)}{Pr(X=x)}\\
&=\frac{Pr(B=b)\mathbf{1}\{x-s\le b \le x\}}{Pr(x-s\le B\le x)}.
\end{align*}
Therefore, with $\mathcal L(\cdot \mid \cdot)$ denoting a conditional distribution,
\[
    \mathcal L(B\mid X=x)
    =\mathcal L\bigl(B\mid x-s\leq B\leq x\bigr),
    \tag{A.5}\label{A5}
\]
For every $x\in J$, the conditioning set in \eqref{A5} contains
$[\mu-t,\mu+t]$. Its omitted probability, $Pr(C_x^c)$ for $C_x=\{x-s\le B\le x\}$, is thus at most
$\varepsilon_n$.
More generally, suppose event $C$ satisfies $q=\Pr(C^c)$. Writing
$Z=B-\mu$ and using $|Z|\leq n$, $\mathbb E [Z]=0$,  and $\mathbb E[Z^2]=\frac{n}{4}$, 
\begin{align*}
    \mathbb E[Z^2\mid C]
    &=\frac{n/4-\mathbb{E}[Z^2\mathbf{1}\{C^c\}]}{1-q}\geq \frac{n}{4}-\frac{n^2q}{1-q},\\
    \bigl|\mathbb E[Z\mid C]\bigr|
    &=\frac{\bigl|\mathbb E[Z\mathbf 1\{C^c\}]\bigr|}{1-q}
      \leq \frac{nq}{1-q}.
\end{align*}
Hence
\[
    \operatorname{Var}(B\mid C)
    \geq
    \frac{n}{4}
    -n^2\left\{\frac{q}{1-q}
                  +\frac{q^2}{(1-q)^2}\right\}.
    \tag{A.6}\label{A6}
\]
Applying \eqref{A6} with the conditioning event in \eqref{A5} and
$q\leq\varepsilon_n$ yields, uniformly over $x\in J$,
\[
    \operatorname{Var}(B\mid X=x)\geq \frac{n}{4}-o(n).
    \tag{A.7}\label{A7}
\]

Finally, if $2t\leq A\leq s-2t$ and $|B-\mu|\leq t$, then $X\in J$.
Independence therefore gives
\[
    \Pr(X\in J)
    \geq
    \left(1-\frac{4t+2}{s+1}\right)(1-\varepsilon_n)
    =1-o(1),
    \tag{A.8}\label{A8}
\]
because $t=O(\sqrt{N\log N})$ while $s\asymp N^{3/4}$.
Equations \eqref{A7}--\eqref{A8}, together with the nonnegativity of conditional variances,
give the lower bound in \eqref{A4}; the reverse bound follows from
$\mathbb E[\operatorname{Var}(B\mid X)]\leq\operatorname{Var}(B)=n/4$.

\end{proof}

\end{document}